\documentclass[12pt]{article}
\usepackage{amsmath, amsthm, amssymb}
\usepackage{booktabs}
\usepackage{caption}

\newtheorem{thm}{Theorem}[section]
\newtheorem{cor}[thm]{Corollary}

\newtheorem{prop}[thm]{Proposition}

\theoremstyle{definition}

\newtheorem{exmp}{Example}[section]

\newcommand{\bit}{\begin{itemize}}
\newcommand{\eit}{\end{itemize}}
\newcommand{\benum}{\begin{enumerate}}
\newcommand{\eenum}{\end{enumerate}}

\newcommand{\barr}{\begin{array}}
\newcommand{\earr}{\end{array}}
\newcommand{\bne}{\begin{equation}}
\newcommand{\ene}{\end{equation}}
\newcommand{\bea}{\begin{eqnarray}}
\newcommand{\eea}{\end{eqnarray}}
\newcommand{\bean}{\begin{eqnarray*}}
\newcommand{\eean}{\end{eqnarray*}}

\begin{document}

\title{Sustainability in the Stochastic Ramsey Model}
\author{Rabi Bhattacharya\\ University of Arizona, Tucson, USA
  \\
  Hyeonju Kim\\ University of Arizona, Tucson, USA   
  \\
Mukul Majumdar\\ Cornell University, Ithaca, USA}

\maketitle

\begin{abstract}
\noindent In this paper we provide a self-contained exposition of the problem of sustaining a constant consumption level in a Ramsey model.  Our focus is on the case in which the output capital-ratio is random.  After a brief review of the known results on the probabilities of sustaining a target consumption from an initial stock, we present some new results on estimating the probabilities by using Chebyshev inequalities.  Some numerical calculations for these estimates are also provided.
\end{abstract}

\section{Introduction}

The discrete time one-good model with a linear production function (``Ramsey Model" in
Dorfman-Samuelson-Solow \cite[Chapter 11.2]{dorf} or McFadden \cite[Section 6]{mcfad})
has long been a convenient framework for exploring many themes in
intertemporal economics.  In this paper, the model is used to throw light on
issues related to \textit{sustainable consumption}.  First, let us pose the
sustainability problem in the deterministic case.  The economy starts with a
positive initial stock $x$ of a good (any reproducible resource or asset: the metaphorical ``corn" of growth theory).  From this a positive quantity $c$
is subtracted.  The parameter $c$ is a datum: it is a target consumption
level that the economy wishes to \textit{sustain}.  If the remainder $i=x-c$
is zero or negative, the economy is ``ruined."  If the remainder is strictly
positive, it is interpreted as an \textit{input} into some productive
activity (i.e., an ``investment").  The \textit{output} of this activity (or, the \textit{returns} from the investment) is then the stock at the beginning
of the next period and is given by $X_{1}=r\cdot i=r\cdot (x-c)$, where $r>0$
is \textit{also} a parameter (``output-capital ratio" in the literature on
planning, or an index of ``productivity" of investment).  Again, in period
one, the parameter $c$ is subtracted from $X_{1},$ and the story is
repeated.  Let $N$ be the first period, if any, such that $X_{N}<0$.  If $N$ is finite, we say that the economy can \textit{sustain} $c$ up to
(but not including) the period $N$ (or, that the economy \textit{survives} up to period $N$).  If $N$ is infinite (i.e., $X_{n}\geq 0$ for all $n$), we say that the consumption target $c$ is \textit{sustainable} (or, the economy \textit{survives} forever).  There are other interpretations of the model.  For example, at a microeconomic level, an economic agent or unit (an investor, a gambler, a firm engaged in managing a fishery,...) is ruined (goes bankrupt, loses the privilege of participating in a game of chance, faces a problem of extinction of the resource managed,...) if its wealth (or, cash reserve, or the stock of the renewable resource,...) $X_{n}$ in
any period falls below some prescribed level $c$ (a minimal rate of
dividend, the fee to participate in the game of chance, the target level of
harvesting,...).  The objective is to study conditions on the parameters $r$, $x$ and $c$ that determine sustainability.  It is not difficult to see that if $r\leq 1,$ then no initial $x$ can sustain any $c>0.$  If $r>1,$ the
economy can sustain $c>0$ if and only if $x\geq  [r/(r-1)]c.$\\

\noindent To extend the scope of our analysis, suppose that the returns from the
investment are uncertain rather than deterministic.  We model this by
introducing an i.i.d. sequence $\epsilon _n$ of positive random
variables.  An investment $i_{n}$ generates output $X_{n+1}$ according to
the rule $X_{n+1}=(\epsilon_{n+1}$)$\cdot i_{n}$.  As in the
deterministic case, the economy starts with an initial stock $x,$ and has a
target consumption $c$.  It is ruined if $(x-c)\leq 0$.  If $x-c>0,$ then the
stock in period one is $X_{1}=\epsilon_{1}(x-c)$.  Again, if $X_{1}-c\leq 0$, it is ruined.  Otherwise, after consumption, $X_{1}-c$ is
invested to generate $X_{2}=\epsilon _{2}\cdot(X_{1}-c$).  In
general, one studies the process;
\begin{equation}
X_{0}=x,\, X_{n+1}=(\epsilon_{n+1})(X_{n}-c)_{+},\text{ where }
a_{+}=max(a,0).  \label{1.1}
\end{equation}

\noindent If $x>c,$ the probability of sustaining $c$ is defined as
\begin{equation}
\rho (x)=P(X_{n}>c\text{ for all }n\geq 0|X_{0}=x).  \label{1.2}
\end{equation}

\noindent It is shown that 
\begin{equation}
\rho (x)=P\left\{ \sum_{n=1}^{\infty}(\epsilon _{1}\epsilon
_{2}..\epsilon_{n})^{-1}<(x/c)-1.\right\}  \label{1.3}
\end{equation}

\noindent This formula (\ref{1.3}) can be used to identify conditions on the common
distribution of $\epsilon_{n}$ under which the value of $\rho (x)$ can
be specified (see Propositions \ref{prop0} and \ref{prop2}).  For example, if $Elog\epsilon_{1}\leq 0$, then $\rho (x)=0$ for all $x$ and $c$.   The case $
Elog\epsilon_{1}>0$ is perhaps the most interesting and turns out to be
challenging.  Note if we define the random variable $Z$ as:
\begin{equation}
Z=\sum_{n=1}^{\infty }(\epsilon_{1}\epsilon_{2}\cdots\epsilon_{n})^{-1}, \label{1.4}
\end{equation}

\noindent we realize that the distribution of $Z$ is crucial in determining $\rho (x)$.  To this effect, we derive a recursive relation that facilitates computing
the moments of $Z$ (Proposition \ref{prop2} and its corollaries).  Next, one obtains
estimates of survival and ruin probabilities by using Chebyshev's
inequalities (section \ref{sec:rho}).  We also address the question of estimating the
probabilities of sustaining a consumption target up to a finite $N$ (see
section \ref{sec:finite}).  Numerical calculations for these estimates are provided in
sections \ref{sec:num} and \ref{sec:finite}.\\

\noindent Our exposition draws upon Majumdar and Radner \cite{majum1992} and Bhattacharya and
Waymire \cite{rabi-way}.  There is a substantial literature using continuous time
models that deals with closely related issues.  Majumdar and Radner \cite{majum1991} derived the survival probability of an agent in a diffusion model and also
extended the analysis to the case in which the agent can sequentially choose
from a set of available technologies.  Radner \cite{radner1998} provided a review of
subsequent research on survival of firms.\\

\noindent Turning to models of mathematical biology, the problem of ruin (extinction)
has been investigated in a variety of contexts (see Brauer and
Castillo-Chavez \cite[Chapters 1, 2]{brau}).  A particularly celebrated example of
``constant yield harvesting" from a population the growth of which is
governed by the logistic law leads to the differential equation;
\begin{equation}
dx/dt=\theta x\left( 1-x/K\right) -c,  \label{1.5}
\end{equation}

\noindent where $\theta >0$ is the ``intrinsic" growth rate, and $K$ is the ``carrying
capacity of the environment" (or, the maximum population size that can be
sustained by the environment.  A complete treatment of the extinction and
sustainability is available (Brauer and Castillo-Chavez \cite[pp.28-29]{brau}).

\section{Remarks on the Deterministic Case}

We make a few remarks on the deterministic case and state the main results.  Here, starting with an initial $x>0,$ the economy is ruined if $x-c\leq 0.$  If $x>c,$ the investment $i_{0}=x-c$ generates the stock $X_{1}=r\cdot i_{0}=r(x-c)$ at the beginning of period 1.  The economy is ruined in
period 1 if $X_{1}-c\leq 0.$  If $X_{1}>c,$ the investment $i_{1}$ =$X_{1}-c$
generates $X_{2}=r\cdot i_{1}=r(X_{1}-c)$ and so on.  If the economy can
sustain $c$ up to (but not including) period $2,$ we know that 
\begin{eqnarray}
c+i_{0} &=&x,  \nonumber \\
c+i_{1} &=&r\cdot i_{0}. \label{2.1}
\end{eqnarray}

\noindent It follows that $c(1+1/r)+i_{1}/r=x,$ leading to:
\bean
c(1+1/r)<x. 
\eean

\noindent Hence, if the economy can sustain $c$ up to period $N$, we must have
\begin{equation}
\sum_{n=0}^{N-1}(1/r^{n})<x/c.  \label{2.2}
\end{equation}

\noindent We immediately conclude that if $r\leq 1,$ then for any $c>0$, there is no $x>0$ such that $c$ can be sustained forever.  Indeed, it is also
easy to verify the following:

\begin{prop}Let $r>1$.  The economy can sustain $c>0$ if and only if 
\begin{equation}
r/(r-1)\leq x/c.  \label{2.3}
\end{equation}
\end{prop}

\section{The Stochastic Ramsey Model}

\subsection{Infinite-Horizon Survival Probability and Conditions on the Common Distribution of $\mathbf{\epsilon_n}$}  

To avoid undue repetition, the economy starts with an initial stock $x>0$ and plans to sustain a consumption level $c>0$.  If $x-c\leq 0,$ it is ruined
``immediately."  So focus on the case $x>c.$  The investment $i_{0}=x-c$ gives
rise to the stock in period 1, $X_{1}=\epsilon_{1}(x-c)$, where $
\epsilon_{1}$ is a nonnegative random variable.  If $X_{1}\leq c,$ the
economy is ruined in period 1; otherwise, if $X_{1}>c,$ the investment
generates the stock in the next period, and so on.  In general, we study:
%



%
%
\bea \label{model1}X_0=x,\quad X_{n+1}=\epsilon_{n+1}(X_n-c)_+\quad
(n\ge 0),\quad a_+=\max(a,0), \eea

\noindent where $\{\epsilon_n:\,n\ge 1\}$ is an i.i.d. sequence of nonnegative
random variables.  The state space may be taken to be $[0,\infty)$
with \emph{absorption} at $0$.  The \emph{probability of survival} of
the economic agent with an initial stock $x>c$ is
\bea\label{def-rho} \rho(x):=P(X_n>c\,\,\text{for all}\,\,n\ge 0|X_0=x). \eea

\noindent Suppose $P(\epsilon_1>0)=1$.  For otherwise, it is simple to check
that that eventual ruin is certain, i.e. $\rho(x)=0$.  From (\ref{model1}), successive iteration yields
\bean
X_{n+1}>c
\,\,&\text{iff}\,\,X_n>c+\frac{c}{\epsilon_{n+1}}\,\,\text{iff}\,\,X_{n-1}>c+\frac{c+c/\epsilon_{n+1}}{\epsilon_n}\cdots\\
&\text{iff}\,\, X_0\equiv x>c+\frac{c}{\epsilon_1}+\frac{c}{\epsilon_1\epsilon_2}+\cdots+\frac{c}{\epsilon_1\epsilon_2\cdots\epsilon_{n+1}}.
\eean
Hence, on the set $\{\epsilon_n>0$ for all $n\}$,
\bean
\{X_n>c\,\,\text{for
  all}\,n\}&=&\big\{x>c+c\sum_{j=1}^{n}\frac{1}{\epsilon_1\epsilon_2\cdots\epsilon_j}\,\,\text{for all}\,\,n
\big\}\\
&=&\big\{x>
c+c\sum_{n=1}^{\infty}\frac{1}{\epsilon_1\epsilon_2\cdots\epsilon_n}
\big\}=\big\{\sum_{n=1}^{\infty}\frac{1}{\epsilon_1\epsilon_2\cdots\epsilon_n}<
\frac{x}{c}-1\big\}
\eean
In other words, 
\bea \label{rho}\rho(x)=P\big\{\sum_{n=1}^{\infty}\frac{1}{\epsilon_1\epsilon_2\cdots\epsilon_n}<
\frac{x}{c}-1\big\}. \eea

\noindent For completeness, we review conditions on the common
distribution of $\epsilon_n$ under which one has (a) $\rho(x)=0$,  (b)$\rho(x)=1$, or (c)
$\rho(x)<1\,\,(x>c)$.  The Strong
Law of Large Numbers gives

\bean \frac{1}{n}\sum_{r=1}^n\log\epsilon_r
\stackrel{a.s.}{\longrightarrow}E\log\epsilon_1.\eean

\noindent Thus, if $E\log \epsilon_1<0$, $\epsilon_1\epsilon_2\cdots\epsilon_n\rightarrow 0$
a.s.  This implies that the infinite series in (\ref{rho}) diverges
a.s., i.e.,
\bea \label{rho.c} \rho(x)=0\quad \text{for all}\,\,x \,\,\text{if}\,\,
E\log\epsilon_1< 0. \eea

\noindent If $\epsilon_1$ is non-degenerate and $E\log\epsilon_1=0$, then
$\sum_{r=1}^n\log\epsilon_r$ also has a subsequence converging to
$-\infty$ a.s., and again the series in (\ref{rho}) diverges and
$\rho(x)=0$, $\forall x>c$.  By  Jensen's Inequality, $E\log\epsilon_1\le \log E\epsilon_1$, with
strict inequality if $\epsilon_1$ is nondegenerate, which we assume.  Therefore, $E\epsilon_1\le1$ implies 
$E\log\epsilon_1<0$, so that $\rho(x)=0$.  Next, let us consider the
case $E\log\epsilon_1>0$.  Define $m:=\inf\{ z\ge
0:\,P(\epsilon_1\le z)> 0\}$.  We will show that

\bea \label{rho.c1} 
\rho(x)<1\quad \text{for all}\,\,x,\,\text{if}\,\,m\le 1. \eea

\noindent Fix $A>0$, however large.  One can find $n_0$ such that
$n_0>A\prod_{r=1}^{\infty}(1+r^{-2})$ as $\prod(1+r^{-2})<\exp\{\sum r^{-2}\}<\infty$.  If $m\le 1$ then
$P(\epsilon_1\le 1+r^{-2})>0$ for all $r\ge 1$.  Consequently,

\bean
0&<&P(\epsilon_r\le 1+r^{-2}\,\,\text{for}\,\,1\le r\le n_0)\\
&\le&P\big(\sum_{r=1}^{n_0}\frac{1}{\epsilon_1\epsilon_2\cdots\epsilon_r}\ge\sum_{r=1}^{n_0}\frac{1}{\prod_{j=1}^{r}(1+1/j^2)}\big)\\
&\le&P\big(\sum_{r=1}^{n_0}\frac{1}{\epsilon_1\epsilon_2\cdots\epsilon_r}\ge
\frac{n_0}{\prod_{j=1}^{\infty}(1+1/j^2)}\big)\\
&\le&P(\sum_{r=1}^{n_0}(\epsilon_1\epsilon_2\cdots\epsilon_n)^{-1}>A)\le P (\sum_{r=1}^{\infty}(\epsilon_1\epsilon_2\cdots\epsilon_r)^{-1}>A).
\eean

\noindent Since $A$ is arbitrary, (\ref{rho}) is less than 1 for all $x$, proving (\ref{rho.c1}).  \\
One may also prove that, for $m>1$,
\begin{equation}\label{rho1}
    \rho(x) \left\{
    \begin{array}{rl}
      <1 & \text{if } x < c(\frac{m}{m-1}),\\ 
      =1& \text{if } x \ge c(\frac{m}{m-1}).\qquad (m>1) .
    \end{array}\right.
\end{equation}
Observe that
$\sum_{n=1}^{\infty}(\epsilon_1\epsilon_2\cdots\epsilon_n)^{-1}\le
\sum_{n=1}^{\infty}m^{-n}=1/(m-1)$, with probability 1 if $m>1$.
 The second relation in (\ref{rho1}) is subsequently drawn by
 (\ref{rho}).  For the first relation in (\ref{rho1}) to be shown, letting
$x<cm/(m-1)-c\delta$ for some $\delta>0$ implies $x/c-1<1/(m-1)-\delta$.  One can choose $n(\delta)$ such that
$\sum_{r=n(\delta)}^{\infty}m^{-r}<\delta/2$ and then choose
$\delta_r>0$ $(1\le r\le n(\delta)-1)$ such that
\bean
\sum_{r=1}^{n(\delta)-1}\frac{1}{(m+\delta_1)\cdots(m+\delta_r)}>\sum_{r=1}^{n(\delta)-1}\frac{1}{m^r}-\frac{\delta}{2}.
\eean
Then
\bean
0&<&P(\epsilon_r<m+\delta_r\text{ for }1\le r\le n(\delta)-1)\le
P\big(
\sum_{r=1}^{n(\delta)-1}\frac{1}{\epsilon_1\cdots\epsilon_r}>\sum_{r=1}^{n(\delta)-1}\frac{1}{m^r}-\frac{\delta}{2}\big)\\
&\le &P\big(
\sum_{r=1}^{\infty}\frac{1}{\epsilon_1\cdots\epsilon_r}>\sum_{r=1}^{\infty}\frac{1}{m^r}-\delta\big)=P\big(
\sum_{r=1}^{\infty}\frac{1}{\epsilon_1\cdots\epsilon_r}>\frac{1}{m-1}-\delta\big).
\eean
For $\delta>0$ small enough, the last probability is smaller than
$P(\sum (\epsilon_1\cdots\epsilon_r)^{-1}>x/c-1)$ if $x/c-1<1/(m-1)$,
i.e., if $x<cm/(m-1)$.  For such $x$, $1-\rho(x)>0$.  The desired
result is obtained.  The following proposition summarizes the above results:\\

\begin{prop}\label{prop0}(\cite{rabi-way}, \cite{majum1992}) Let $m:=\inf\{ z\ge
0:\,P(\epsilon_1\le z)> 0\}$.
\bit
\item[(a)] If $E\log\epsilon_1 \le 0$, then $\rho(x)=0$ for all $x$ and $c$.\\
\item[(b)] If $E\log\epsilon_1>0$, then 
\begin{equation*}
    \rho(x) =
    \begin{cases}
     <1& \text{if } m\le 1 \,\,\forall x, \text{ or } x<c\frac{m}{m-1}\,\, (m>1),\\
      =1 & \text{if } x\ge c\frac{m}{m-1}\,\,(m>1).
    \end{cases}
\end{equation*}

\eit
\end{prop}

\noindent The subsequent Proposition~\ref{prop1} provides more explicit
statements on $\rho(x)$.  This allows $\rho(x)$ to be constructed by
estimating $Z$ with the common distribution
of $\epsilon_n$. \\

\begin{prop}\label{prop1} (\cite{rabi-way}, \cite{majum1992})
Assume $E\log\epsilon_1>0$, $Z:= \sum_{1\le
   n<\infty}(\epsilon_1\epsilon_2\cdots\epsilon_n)^{-1}$.  
Define $d_1= \inf\{z\ge 0: P(Z\le z) >0\},$ $d_2 =
\sup\{z\ge 0: P(Z\ge z)>0\}$, and $M= \sup\{z\ge
0: P(\epsilon_1\ge z)>0\}$.  Then,
\bit
\item[(a)] $Z$ is finite (almost surely).
\item[(b)]     
\begin{equation*}\rho(x) = \left\{
    \barr {rl}
      0 & \text{if  } x< c(d_1+1),\\
      \in (0,1)& \text{if  } c(d_1+1)<x< c(d_2+1),\\
      1 & \text{if  } x>c(d_2+1).\\
    \earr \right.
\end{equation*}        
\item[(c)] $\rho(x) =0$ if  $x < \frac{cM}{M-1}$ ($1<M<\infty$), or for all
  $x$ ($M\le 1$).
\item[(d)] One can express the (essential) lower bound $d_1$ and upper bound $d_2$ of $Z$ in terms of those of $\epsilon_1$, namely, $m$ and $M$:\bit
\item [(i)] $d_1= \sum_{1\le n<\infty} M^{-n} = 1/(M-1)$ if $M > 1$, and
$d_1 = \infty$ if $M \le 1$.  
\item[(ii)]  $d_2= \sum_{1\le n<\infty} m^{-n} = 1/( m-1)$ if $m > 1$, and
$d_2 = \infty$ if $m \le 1$, where $m$ is defined in (\ref{rho.c1}).   
\eit
\eit
\end{prop}
\begin{proof} \emph{(a)}  By the Strong Law of Large Numbers, $ (\sum_{1\le j\le n} \log
  \epsilon_j)/n \rightarrow \mu$ (with probability 1), where $\mu = E \log \epsilon_1 (>0)$.  Therefore,
  there exists a random variable $N$ which is finite a.s. such that
  $(\sum_{1\le j\le n} \log \epsilon_j)/n > \mu/2$ for all $n > N$.
  In other words,
  $(\epsilon_1\epsilon_2\cdots\epsilon_n)^{-1} <  e^{ -n\mu/2}$ for $n > N$.
This suggests that 
\bea 
Z &=&\sum_{1\le n\le N}
(\epsilon_1\epsilon_2\cdots\epsilon_n)^{-1}+\sum_{n>N}(\epsilon_1\epsilon_2\cdots\epsilon_n)^{-1} \\
&<&\sum_{1\le n\le N} (\epsilon_1\epsilon_2\cdots\epsilon_n)^{-1} +  \sum_{n>N}
e^{-nμ/2} <\infty\,\,a.s.
\eea

\emph{(b)}  $x < c(d_1+1)$ implies $x/c-1 < d_1$.  One can find
$\theta>0$ such that $x/c-1 < d_1-\theta$, which implies 
 \bean \rho(x) =  P(Z \le x/c-1) \le  P( Z < d_1-\theta) =0.\eean
Likewise, $x > c(d_2+1)$ indicates $x/c -1 > d_2$.  Again, one can find
$\theta>0$ such that $x/c-1 > d_2+\theta$, which indicates 
\bean 1= P(Z < d_2+\theta) \le  P(Z \le x/c-1) = \rho(x).\eean
Finally, $c(d_1+1) < x<c(d_2+1)$ suggests $d_1 < x/c-1 < d_2$.  Then
$x/c-1 > d_1+ \theta$ for some $\theta>0$, which suggests $\rho(x) =
P(Z \le  x/c-1)\ge P(Z < d_1+\theta) >0$.  Similarly, $x/c-1 < d_2
-\theta'$ for some $\theta'>0$, which turns out to be $\rho(x) =  P(Z
\le x/c-1) \le P(Z< d_2 -\theta') <1$. \\

\emph{(c)} $x < cM/(M-1)$ can be rewritten in the form of $x/c -1 < 1/(M-1)$.
Notice that $P( \epsilon_1 > M) =0$.  This is because if $P( \epsilon_1 > M) >0$, then
$P(\epsilon_1 \le M) < 1$.  Thus, there exists $\theta>0$ such that
$P(\epsilon_1\ge  M+\theta) >0$, contradiction.  $Z= \sum_{1\le n<\infty} (\epsilon_1\epsilon_2\cdots\epsilon_n)^{-1} \ge
\sum_{1\le n<\infty} M^{-n} = 1/(M-1)>x/c -1$, so that  $\rho(x) =
P(Z \le x/c-1) =0$.\\  

\noindent Next, $M\le 1$ leads to $\epsilon_1\epsilon_2\cdots\epsilon_n \le 1$
for all $n$.  Then, $(\epsilon_1\epsilon_2\cdots\epsilon_n)^{-1} \ge
1$ for all $n$, which implies $Z = \infty$ almost surely, and $\rho(x)
=  P(Z \le x/c -1) =0$, no matter how large $x$ may be. \\

\emph{(d)-(i)} For $M \le 1$, $P(\epsilon_n \le M) =1$ for all $n$, yielding $Z = \sum_{1\le
   n<\infty} (\epsilon_1\epsilon_2\cdots\epsilon_n)^{-1} \ge
 \sum_{1\le n<\infty} M^{-n} =\infty$ almost surely.  It follows that $d_1 =\infty$.\\

\noindent For $M >1$, again, $P (\epsilon_n
   \le M) =1$ for all $n$, suggesting $Z = \sum_{1\le n<\infty}
   (\epsilon_1\epsilon_2\cdots\epsilon_n)^{-1} \ge \sum_{1\le
     n<\infty} M^{-n}  = 1/(M-1)$, almost surely.  Therefore, $d_1\ge 1/(M-1)$.   To prove $d_1\le 1/(M-1)$, note that there
   exists $\theta>0$ such that  $M -\theta >1$, and $P(\epsilon_1 >
   M-\theta)>0$ by the definition of $M$.  Since $\epsilon_n$'s are
   independent, $P(\epsilon_n>M-\theta\text{ for all }
   n=1,2,...,N)=\prod_{1\le n\le N}P(\epsilon_n>M-\theta)>0$ for every
   $N$.  This implies $P(\sum_{1\le n\le N} (\epsilon_1\epsilon_2\cdots\epsilon_n)^{-1} <
   \sum_{1\le n\le N} (M-\theta)^{-n} ) > 0$ for every $N$.  Besides,
   $\sum_{1\le n\le N}
   (\epsilon_1\epsilon_2\cdots\epsilon_n)^{-1}\rightarrow Z$, and
   $\sum_{1\le n\le N}(M-\theta)^{-n}$ converges to $1/(M-\theta-1)$ as $N\rightarrow
   \infty$.  It turns out that $P( Z \le 1/(M-\theta-1)) >0$.  Therefore,
   $d_1 \le 1/(M-\theta-1)$ for every $\theta > 0$.  Letting $\theta
   \downarrow 0$ gives rise to $d_1 \le 1/(M-1)$.\\

\emph{(d)-(ii)} For $m > 1$, $P (\epsilon_1 \ge m) =1$, indicating $Z =\sum_{1\le n<\infty}
(\epsilon_1\epsilon_2\cdots\epsilon_n)^{-1} \le \sum_{1\le n<\infty}
m^{-n} =1/(m-1)$ almost surely.  It follows $d_2 \le 1/(m-1)$.  Note
that $P( Z \ge 1/(m-1) +\theta') =0$ for all $\theta' > 0$.  To prove
$d_2 \ge 1/(m-1)$, one obtains $P( \epsilon_1 < m+\theta) > 0$ for any
$\theta>0$ and by the definition of $m$.  Arguing as in \emph{(i)}, one
demonstrates that $P(\sum_{1\le n\le N} (\epsilon_1\epsilon_2\cdots\epsilon_n)^{-1}
>\sum_{1\le n\le N} (m+\theta)^{-n} ) > 0$ for every $N$, and $P(Z \ge
1/(m+\theta-1) ) > 0$ as $N\rightarrow \infty$.  This proves that $d_2
\ge 1/(m+\theta-1)$.  This is true for every $\theta > 0$, so that $d_2 \ge 1/(m-1)$.\\ 

\noindent Now let $m \le 1$.  For every $\theta>0$, $P (\epsilon_1 \le
 1+\theta) >0$, implying $P(Z\ge \sum_{1\le n\le N}(1+\theta)^{-n} ) > 0$.
 Since $Z=\sum_{1\le
   n<\infty}(\epsilon_1\epsilon_2\cdots\epsilon_n)^{-1} ≥\ge
  \sum_{1\le n\le N}
   (1+\theta)^{-n} \rightarrow 1/\theta$ as $N\rightarrow \infty$.  Hence, $P(Z \ge
   1/\theta)  > 0$ for every $\theta > 0$, implying $d_2 = \infty$.
\end{proof}

\subsection{Recursive Computation of the Moments of $Z$} \label{sec:Z}

The following novel method using a recursive relation facilitates computing the moments of $Z$ whose distribution is quite intractable.\\

\begin{prop}\label{prop2} One has the relation
\bea   Z = (1/\epsilon_1)( 1+ W), \eea 
where $W =  \sum_{2\le n\le N}
(\epsilon_2\epsilon_3\cdots\epsilon_n)^{-1}$ has the same distribution
as $Z$, and $W$ and $\epsilon_1$ are independent.  
\end{prop}

\begin{cor}\label{cor2} Let $E(\log \epsilon_1) > 0$.  (i) If $m=0$, then $d_2= \infty$, and (ii) if $M=\infty$, then $d_1 =0$.  In both cases, $0<\rho(x)<1$ for $\forall x>c$. \end{cor}
\begin{proof}
(i) $Z > 1/\epsilon_1$, which exceeds any large value with positive
probability.  (ii) $(1/\epsilon_1)( 1+ W)$ approaches zero as
$\epsilon_1$ goes to infinity.  Now we use Proposition~\ref{prop1}.
\end{proof}

\begin{cor} \label{cor3} Let $E(\log \epsilon_1) > 0$.  Denote the moments of $Z$
  and $1/\epsilon_1$ by $\beta_r= EZ^r$, $\gamma_r=
  E(1/\epsilon_1)^r$, respectively $(r=1,2,...)$.  Then, for all $r$ such that $\gamma_r<1$,
\bea   
 \beta_r&= & \gamma_r \sum_{0≤\le j\le r}\binom{r}{j}\beta_j;\,\, (1-\gamma_r)\beta_r=  \gamma_r \sum_{0\le
  j\le r-1} \binom{r}{j}\beta_j;\nonumber \\
  \beta_r&= & [\gamma_r /(1-\gamma_r)]\sum_{0\le j\le r-1}
  \binom{r}{j}\beta_j.    
\eea
If $\gamma_r\ge 1$ for some $r$, $\beta_r=\infty$.
\end{cor}
\begin{proof}
This relation is derived directly from the
  representation in Proposition~\ref{prop2} by independence of  $W$ and $\epsilon_1$ and by the binomial
  formula for $(1+W)^r$. 
\end{proof}


\begin{exmp}\label{ex1} \emph{(Lognormal)} Let $\epsilon_1=e^N $ be lognormal,
  where $N$ is a Normal random variable with mean $\mu > 0$, and a
  positive variance $\sigma^2$.  Applying Corollary~\ref{cor2} gives $d_1 =0$ and $d_2=
  \infty$.  Hence, by Proposition~\ref{prop1},  $0 < \rho(x) < 1$ for every $x >c$.
$1/\epsilon_1 = e^{-N}$ is also lognormal, $-N$
  being Normal with mean $-\mu <0$ and variance $\sigma^2$.  Hence, the moments of $1/\epsilon_1$ are given by  
\bea  \label{rth.z1}\gamma_r= E\epsilon_1^{-r}  = e^{-r\mu + \frac{r^2\sigma^2}{2}}\,\, (r=1,2,...).\eea
By Corollary~\ref{cor3} and (\ref{rth.z1}), the moments $\beta_r$ of $Z$ may now
be computed.  One must require $r<2\mu/\sigma^2$.\\

\noindent\textbf{NOTE:} If $N$ is Normal with mean $\mu$ and variance $\sigma^2$, then the rth moment of $e^N$ is $E(e^N)^r  =    E (e^{rN}) = e^{r\mu +(1/2)r^2\sigma^2}$, from the well known formula for the moment generating function of the Normal distribution.
\end{exmp} 

\begin{exmp}\label{ex2}\emph{(Pareto)} Let $k>0$, $\beta>0$.  Suppose
  $\epsilon_1\sim$Pareto($\beta, k)$ with density $f(x)=\beta
k^{\beta}/x^{\beta+1}\mathbb{I}_{\{x\ge k\}}$.  Let $k$ and $\beta$ be
such that $E\log\epsilon_1=\log k+\beta^{-1}>0$.  Then, one obtains
$0<\rho(x)<1$ for every $x>c$ for $e^{-1/\beta} <k\le 1$.  Now the density function of $1/\epsilon_1$ is $f(y)=\beta
k^{\beta}y^{\beta-1}\mathbb{I}_{\{0<y\le \frac{1}{k}\}}$, and its $r$th
moment is provided by 
\bea E\epsilon_1^{-r}=\frac{\beta}{k^{r}(\beta+r)}.  \eea
Again, the moments of $Z$ can be obtained in the same fashion as in Example~\ref{ex1}.
\end{exmp}  

\begin{exmp}\emph{(Gamma)} Assume $\epsilon_1\sim$ Gamma($\alpha,
  \theta$) with density 
  \bean 
  f(x)=\frac{\theta^{\alpha}}{\Gamma(\alpha)} x^{\alpha-1}e^{-\theta
    x}\mathbb{I}_{\{0<x<\infty\}},\,\,\,\alpha,\, \theta>0.
  \eean  
  For $\alpha>\theta$,
  $E\log\epsilon_1=\frac{\theta^{\alpha}}{\Gamma(\alpha)}\int_0^{\infty}(\log
  x)x^{\alpha-1}e^{-\theta x}dx >0$.  By
  Proposition~\ref{prop1}-\emph{(b),(d)}, $d_1=0$, $d_2=\infty$,
  and subsequently $0<\rho(x)<1$ for every $x>c$.  The density
  function of $1/\epsilon_1$ is
  $f(y)=\frac{\theta^{\alpha}}{\Gamma(\alpha)}y^{-\alpha-1}e^{-\theta/y}\mathbb{I}_{\{0<y<\infty\}}$,
which is an Inverse-Gamma($\alpha,\frac{1}{\theta}$).  Since the moment
generating function of the Inverse-Gamma distribution does not exist, one can attain the finite moments by
direct integration:
\bea
E\epsilon_1^{-r}&=&\frac{\theta^r\Gamma(\alpha-r)}{\Gamma(\alpha)},\text{
  or}\\
                          &=&\frac{\theta^r}{(\alpha-1)(\alpha-2)\cdots(\alpha-r)}
                          \text{ for }\alpha\in\mathbb{Z}^+. 
\eea 
In the same manner, the moments of $Z$ can be obtained.
\end{exmp}

\subsection{Approximation to Survival Probability by Multiple Chebyshev
  Inequalities}\label{sec:rho}

With the moments of $Z$ recursively obtained in Corollary~\ref{cor3}, one obtains conservative estimates of
ruin and survival probabilities using Chebyshev Inequality: 
 \bea \label{phi}
1-\rho(x) =P(Z \ge x/c -1) < \frac{\beta_r }{(x/c-1)^r}\,\,(r=1,2,...),\,\,  (x >c).\eea
Notice that the smaller the upper estimate of ruin probability, the
better the approximation of the true ruin probability is.  Equivalently, the larger the lower estimate of survival
probability, the better.  Therefore, the estimate (\ref{phi}) with $r$ over the
one with $r+1$ should be selected, iff
\bea
\frac{\beta_r}{(x/c-1)^r}\le \frac{\beta_{r+1}}{(x/c-1)^{r+1}},\text{
  or } x\le c(1+\frac{\beta_{r+1}}{\beta_r}).
\eea
The upper estimate of ruin probability and the lower of survival probability, consequently, are 
obtained as follows:
\bea
1-\rho(x)<\frac{\beta_r}{(x/c-1)^r},\quad
\rho(x)>1-\frac{\beta_r}{(x/c-1)^r},\,\,(x>c),\label{lower}
\eea
where $r$ is chosen as follows:\\
\bea
\begin{cases}
c(1+\beta_r/\beta_{r-1})<x\le c(1+\beta_{r+1}/\beta_r) & \text{if } r\ge 2, \nonumber\\
c<x\le c(1+\beta_2/\beta_{1}) & \text{if } r=1, \nonumber\\
\end{cases}
\eea
subject to the restriction $\gamma_r<1$.

\subsection{Numerical Examples}\label{sec:num}

A conservative lower estimate of the survival probability using
Chebyshev inequalities with different orders of moments of $Z$
depending on $x$ in (\ref{rho}) is numerically obtained in
the following two examples.  From these examples, one employs the empirical
cumulative distribution function (ECDF), $\hat{F}(x)$, of $Z=\sum_{1\le
  n<\infty}(\epsilon_1\epsilon_2\cdots\epsilon_n)^{-1}$ to examine the performance of the estimate of survival probability using Chebyshev
inequalities with different orders of moments of $Z$ depending on $x$.  To obtain the empirical estimate of the distribution function of $Z$, one generates $n$ random variables of $\epsilon_1^{-1}$, where
$\epsilon_1$ is distributed by lognormal, or Pareto distribution in the examples.  The sum of cumulative
products of $\epsilon_1^{-1}$'s yields a random variable, $Z$, whose distribution is obtained by replicating 3000 $Z$'s ($N=3000$) in this section and the next.  The
empirical cumulative distribution (ECDF) of
$Z$ is then produced by $F_N(x)=\frac{1}{N}\sum_{j=1}^{N}\mathbb{I}\{Z_{n,j}=\sum_{i=1}^{n}\frac{1}{\epsilon_1\epsilon_2\cdots\epsilon_i}<x/c-1\}=\frac{1}{N}\sum_{j=1}^{N}\mathbb{I}\{c(Z_{n,j}+1)<x\}$ ($N=3000$),
i.e. the proportion of time that the agent survives until
time $n$ provided that the initial
stock is $X_0=x$.  Since this empirical cumulative distribution will be close to the true probability ($\rho(x)$) for large enough
simulations, we label $\hat{F}(x)$ as $\rho(x)$ from now on.  Depending on the number of drawings of
$\epsilon_1^{-1}$, one can obtain either a finite-horizon survival
probability ($\rho_n(x)$) or an
infinite-horizon survival probability ($\rho(x)$).\\
\\

\begin{exmp}\label{ex5}(\emph{Lower estimates ($\underline{\rho(x)}$) of $\rho(x)$: Lognormal
    vs. Pareto.}) 
Table~\ref{tab:prob}\protect\footnote{At least 100 of $\epsilon_1$'s were drawn to have the series for $Z$ converge.  This is because the series for $Z$ in the Pareto case converges more slowly than it does in the lognormal case.} shows survival probabilities ($\rho(x)$) and the corresponding
  lower estimates ($\underline{\rho(x)}$) using multiple-Chebyshev inequalities with different orders of moments of $Z$ depending on $x$, where $\epsilon_1$'s are distributed by lognormal, Pareto, respectively .  Since $\rho(x)$ of Pareto distribution reaches 1 immediately for $k>1$ by Proposition~\ref{prop1}, it is not of interest anymore.  Instead, we only consider the case, where $e^{-1/\beta} <k\le
1$ as described in Example~\ref{ex2}, i.e. $0<\rho(x)<1$.  Now, to compare the lower
estimates of survival probabilities for the two distributions, the first and second moments of
$\epsilon_1^{-1}$ are matched by having the parameters of Pareto
distribution ($\beta, k$) free to select.   \\
\\
Table~\ref{tab:Lognormal}, Table~\ref{tab:Pareto.9} present some moments of $\epsilon_1^{-1}$, and those of $Z$ by recursive
computation in Corollary~\ref{cor3}, and the corresponding boundaries
of $x$ obtained by (\ref{lower}) for lognormal and Pareto distributions of
$\epsilon_1$'s, respectively.  Both finite moments of $Z$,
$\beta_r=EZ^r$, can be computed up to the point that $\gamma_r=E\epsilon_1^{-r}$ does
not exceed 1.  Table~\ref{tab:Lognormal}, Table~\ref{tab:Pareto.9}
suggest that $\beta_r=EZ^r$ possibly decreases as far as
$\gamma_r=E\epsilon_1^{-r}<1$.  However, $\beta_{r+1}/\beta_r$
increases, so that the boundary value increases.  For
$e^{-1/\beta}<k<1$, one can choose $k=0.9$ and $\beta=0.1$ such that
$\epsilon_1\sim$Pareto($\beta=0.1$, $k=0.9$).   The parameters of the lognormal
distribution are therefore taken as $\mu=3.17$ and $\sigma^2=1.75$.
The recursive computation to produce the moments for $Z$ yields the result that for the lognormal case,
$\gamma_4=E\epsilon^{-4}>1$ for the first time, which leads to $EZ^4=\infty$
according to Corollary~\ref{cor3}.  Therefore, the lower estimate of
survival probability using multiple-Chebyshev inequalities with different orders of
moments of $Z$ depending on $x$ can be drawn only by the first and
second moments of $Z$.  Here, $\beta_3=EZ^3$ is employed to get the boundary of $x$ for
$\rho(x)>1-\beta_2/(x/c-1)^2$ in (\ref{lower}).  For the Pareto case,
$\gamma_{61}=E\epsilon^{-61}>1$ for the first time, which results in $\beta_{61}=EZ^{61}=\infty$.
Thus, the complete lower estimate of $\rho(x)$ by multiple-Chebyshev inequalities can be achieved by using 59 moments of $Z$.  Similar to the
lognormal case, $\beta_{60}=EZ^{60}$ is used to attain the boundary of
$x$ for $\rho(x)>1-\beta_{59}/(x/c-1)^{59}$.\\

\noindent Table~\ref{tab:prob} exhibits the percentiles of survival
probabilities and the corresponding lower estimates obtained by multiple-Chebyshev
Inequalities with different orders of $Z$ depending on $x$, where $\epsilon_1$'s are respectively distributed by lognormal and
Pareto distributions.  Since the first and second moments ($EZ$,
$EZ^2$) are matched for both distributions, the lower estimates
obtained by those moments are basically the same as far as $x\le 1.9481$, in
Table~\ref{tab:Pareto.9}.  After that, the lower estimate of $\rho(x)$
for Pareto distribution becomes larger
than that of $\rho(x)$ for lognormal
distribution.  That is because the remaining lower estimates of
$\rho(x)$ for the Pareto case are obtained by Chebyshev inequalities
with higher orders of moments of $Z$ than the lower estimates of $\rho(x)$ for the lognormal case.  
Further, the survival probabilities for the Pareto case are lower than
those for the lognormal case overall as $x$ increases in
Table~\ref{tab:prob}.  One can conclude that the lower estimate ($\underline{\rho(x)}$) by multiple-Chebyshev
inequalities with different orders of moments of $Z$ depending on $x$
for the Pareto distribution case more closely approximates the
corresponding survival probability ($\rho(x)$) than that
of $\rho(x)$ for the lognormal distribution case does as $x$ gets larger. \\

\begin{table}[htb]
\begin{minipage}{.45\textwidth}
\centering
\begin{tabular}{rrrr}
  \toprule
r & $E\epsilon_1^{-r}$ & $EZ^r$ & Boundaries \\ 
  \midrule
1 & 0.1010 & 0.1124 & 1.6808 \\ 
  2 & 0.0588 & 0.0765 & 6.0288 \\ 
  3 & 0.1971 & 0.3847 &   Inf \\ 
   \bottomrule
\end{tabular}\captionof{table}{$\ln\epsilon_1\sim$N(3.17,1.75)} 
\label{tab:Lognormal}
\end{minipage}
\begin{minipage}{.40\textwidth}
\centering
\begin{tabular}{rrrr}
  \toprule
r & $E\epsilon_1^{-r}$ & $EZ^r$ & Boundaries \\ 
  \midrule
1 & 0.1010 & 0.1124 & 1.6808 \\ 
  2 & 0.0588 & 0.0765 & 1.9481 \\ 
  3 & 0.0442 & 0.0725 & 2.1704 \\ 
  4 & 0.0372 & 0.0849 & 2.4067 \\ 
   \bottomrule
\end{tabular}\captionof{table}{$\epsilon_1\sim$ \\ Pareto($\beta$=0.1,$k$=0.9)}
\label{tab:Pareto.9}
\end{minipage}
\end{table}

\begin{table}[ht]
\centering
\begin{tabular}{rlllllll}
  \hline
  \hline
$x$ & 1.1 & 1.2 & 1.4 & 1.6 & 1.8 & 2 & 2.2 \\ 
  lognormal & 0.7193 & 0.8633 & 0.9513 & 0.9777 & 0.9863 & 0.9897 & 0.992 \\ 
  (lognormal) & 0 & 0.4382 & 0.7191 & 0.8127 & 0.8805 & 0.9235 & 0.9469 \\ 
  Pareto & 0.7723 & 0.8267 & 0.8913 & 0.9283 & 0.9553 & 0.9827 & 0.9963 \\ 
  (Pareto) & 0 & 0.4382 & 0.7191 & 0.8127 & 0.8805 & 0.9275 & 0.9591 \\ 
   \hline
\end{tabular}
\caption{Survival probabilities, $\rho(x)$, (lognormal, Pareto) vs lower estimates, \underline{$\rho(x)$}, by multiple-Chebyshev's inequalities with different moments of $Z$ depending on $x$ ((lognormal), (Pareto)) for $\ln\epsilon_1\sim$ N(3.17,1.75), $\epsilon_1\sim$ Pareto($\beta$=0.1,$k$=0.9).} 
\label{tab:prob}
\end{table}

%

\end{exmp}

\subsection{Finite-Horizon Survival Probability and Approximation by Multiple Chebyshev Inequalities}\label{sec:finite}

As in section~\ref{sec:Z} and~\ref{sec:rho}, one can compute the
moments of $Z_n:=\sum_{1\le j\le n}(\epsilon_1\cdots\epsilon_j)^{-1}$
and achieve the conservative estimates
of ruin and survival probabilities in finite time.\\

Consider the derivation of $\rho(x)$ in (\ref{def-rho})-(\ref{rho}).
Then it is not hard to set up the probabilities of survival and ruin
until time $n$.  The survival
probability of an economic agent up to the finite time $n$ with an initial stock $x>c$ is 
\bea
\rho_n(x)&:=&P(X_n>c|X_0=x)
\nonumber\\
               &:=& P(Z_n <x/c-1).
\eea
The finite-horizon ruin probability of an economic agent with an
initial stock $x$ is then $1-\rho_n(x)$.\\
\\
Now write $W_{n-1}:=\sum_{2\le j\le
  n}(\epsilon_1\cdots\epsilon_j)^{-1}$.  As the relation in
Proposition~\ref{prop2}, one can rewrite $Z_n$ in terms of $W_{n-1}$
and $\epsilon_1$:
\bea Z_n\stackrel{\mathcal{L}}{=}\epsilon_1^{-1}(1+W_{n-1}), \eea
where $W_{n-1}$ has the same distribution as $Z_{n-1}$, and
$W_{n-1}$ and $\epsilon_1$ are independent.  With this relation, for $E(\log\epsilon_1)>0$, the moments of $Z_n$
for all $n=1,2,...$ are calculated recursively. 
\bea\label{fin-EZ}
\beta_r^{(n)}=\gamma_rE(1+W_{n-1})^r=\gamma_r\sum_{0\le j\le r}\binom{r}{j}\beta_r^{(n-1)},
\eea
where $\beta_r^{(n)}:=EZ_n^r$, $\beta_0^{(n)}:=1$ for all $ n\ge 1$,
and $\beta_r^{(0)}:=0$ for all $r$.  $\gamma_r:=E(1/\epsilon_1)^r$
$(r=1,2,...)$, and $\gamma_0:=1$.\\

\begin{exmp} According to the relation (\ref{fin-EZ}), one can obtain
  explicit forms of the moments of $Z_n$, $\beta_r^{(n)}=EZ_n^r$, for $n=1,2,...$:
  \bean
  \beta_r^{(1)}&=&\gamma_r\,\,(\forall r=0,1,2,...);\\
  \beta_1^{(2)}&=&\gamma_1(1+\beta_1^{(1)})=\gamma_1(1+\gamma_1),\\
  \beta_2^{(2)}&=&\gamma_2(1+2\beta_1^{(1)}+\beta_2^{(1)})=\gamma_2(1+2\gamma_1+\gamma_2),\\
  \beta_3^{(2)}&=&\gamma_3(1+3\beta_1^{(1)}+3\beta_2^{(1)}+\beta_3^{(1)})=\gamma_3(1+3\gamma_1+3\gamma_2+\gamma_3),...;\\
  \beta_1^{(3)}&=&\gamma_1(1+\beta_1^{(2)})=\gamma_1(1+\gamma_1(1+\gamma_1))=\gamma_1+\gamma_1^2+\gamma_1^3),\\
  \beta_2^{(3)}&=&\gamma_2(1+2\beta_1^{(2)}+\beta_2^{(2)})=\gamma_2(1+2\gamma_1+2\gamma_1^2+\gamma_2+2\gamma_1\gamma_2+\gamma_2^2),...;\\
  \vdots
  \eean
  
\end{exmp}

Finally the upper estimate of ruin probability and the lower of survival
probability until time $n$ are 
obtained as follows:
\bea
1-\rho_n(x)<\frac{\beta_r^{(n)}}{(x/c-1)^r},\quad
\rho_n(x)>1-\frac{\beta_r^{(n)}}{(x/c-1)^r},\,\,(x>c),\label{fin-lower}
\eea
where $r$ is chosen as follows:\\
\bea
\begin{cases}
c(1+\beta_r^{(n)}/\beta_{r-1}^{(n)})<x\le c(1+\beta_{r+1}^{(n)}/\beta_r^{(n)}) & \text{if } r\ge 2, \nonumber\\
c<x\le c(1+\beta_2^{(n)}/\beta_{1}^{(n)}) & \text{if } r=1, \nonumber\\
\end{cases}
\eea
\\
\begin{exmp}\label{ex6}(\emph{Lower estimates ($\underline{\rho_n(x)}$) of finite-horizon survival probabilities ($\rho_n(x)$): Lognormal, Pareto, and Gamma.})
  Table~\ref{tab:bdary.l}-Table~\ref{tab:bdary.g} demonstrate boundary values of $x$ based on
  (\ref{fin-lower}), where $\epsilon_1$'s are distributed by lognormal,
  Pareto, and gamma distributions, respectively.  In the same fashion as in Example~\ref{ex5}, the first and second moments of
  $\epsilon_1^{-1}$'s are matched for all the three
  distributions, so that one can compare the survival probabilities ($\rho_n(x)$),
  with the corresponding lower bounds ($\underline{\rho_n(x)}$) by multiple-Chebyshev
  inequalities with different moments of $Z_n$ depending on $x$.  In
  such a way, one gets $\ln\epsilon_1\sim$ N(0.2146,0.0645),
  $\epsilon_1\sim$ Pareto($\beta$=3,$k$=0.9), and $\epsilon_1\sim$
  Gamma($\alpha$=17,$\theta$=13.3333).  The finite moments of $Z_n$
  are calculated recursively from these distributions.  These tables show that the
  boundaries of $x$ increase for $\forall r\ge 1$ as $n$ increases,
  i.e., $Z_n\longrightarrow Z$.\\

\begin{table}[ht]
\centering
\begin{tabular}{llllll}
  \hline
r & $Z_3$ & $Z_5$ & $Z_{10}$ & $Z_{20}$ & $Z$ \\ 
  \hline
c & 1 & 1 & 1 & 1 & 1 \\ 
  1 & 3.3137 & 4.3807 & 5.9908 & 7.0502 & 7.2857 \\ 
  2 & 3.546 & 4.8419 & 7.0433 & 8.8004 & 9.2795 \\ 
  3 & 3.8072 & 5.3915 & 8.4725 & 11.6162 & 12.7826 \\ 
  4 & 4.1018 & 6.0525 & 10.4814 & 16.7021 & 20.5384 \\ 
  5 & 4.4353 & 6.8551 & 13.4176 & 27.5237 & 52.1729 \\ 
   \hline
\end{tabular}
\caption{Boundaries of $x$ for the lower estimates of $\rho_n(x)$ by multiple-Chebyshev inequalities with $r^{th}$ moments of $Z_n$ $(n=3,5,10,20)$ and $Z$, where $\ln\epsilon_1\sim$ N(0.2146,0.0645), and c=1 (a fixed amount of consumption).} 
\label{tab:bdary.l}
\end{table}
\begin{table}[ht]
\centering
\begin{tabular}{llllll}
  \hline
r & $Z_3$ & $Z_5$ & $Z_{10}$ & $Z_{20}$ & $Z$ \\ 
  \hline
c & 1 & 1 & 1 & 1 & 1 \\ 
  1 & 3.3137 & 4.3807 & 5.9908 & 7.0502 & 7.2857 \\ 
  2 & 3.4698 & 4.6962 & 6.7183 & 8.2603 & 8.6618 \\ 
  3 & 3.5932 & 4.9592 & 7.3887 & 9.5165 & 10.1737 \\ 
  4 & 3.6938 & 5.1826 & 8.0083 & 10.8238 & 11.8661 \\ 
  5 & 3.7777 & 5.3752 & 8.5816 & 12.1805 & 13.791 \\ 
   \hline
\end{tabular}
\caption{Boundaries of $x$ for the lower estimates of $\rho_n(x)$ by multiple-Chebyshev's inequalities with $r^{th}$ moments of $Z_n$ ($n=3,5,10,20$) and $Z$, where $\epsilon_1\sim$ Pareto($\beta$=3,$k$=0.9), and c=1 (a fixed amount of consumption).} 
\label{tab:bdary.p}
\end{table}
\begin{table}[ht]
\centering
\begin{tabular}{llllll}
  \hline
r & $Z_3$ & $Z_5$ & $Z_{10}$ & $Z_{20}$ & $Z$ \\ 
  \hline
c & 1 & 1 & 1 & 1 & 1 \\ 
  1 & 3.3137 & 4.3807 & 5.9908 & 7.0502 & 7.2857 \\ 
  2 & 3.5606 & 4.87 & 7.1073 & 8.9091 & 9.405 \\ 
  3 & 3.8589 & 5.4978 & 8.7526 & 12.201 & 13.546 \\ 
  4 & 4.2255 & 6.3255 & 11.3481 & 19.2233 & 25.1935 \\ 
  5 & 4.6846 & 7.4534 & 15.8266 & 39.6022 & 256.6073 \\ 
   \hline
\end{tabular}
\caption{Boundaries of $x$ for the lower estimates of $\rho_n(x)$ by multiple-Chebyshev inequalities with $r^{th}$ moments of $Z_n$ ($n=3,5,10,20$) and $Z$, where $\epsilon_1\sim$ $\Gamma$($\alpha$=17,$\theta$=13.3333), and c=1 (a fixed amount of consumption) (Computations in Tables \ref{tab:bdary.l}-\ref{tab:bdary.g} corresponding to any c>0 can be handled by interpreting x as x/c.).} 
\label{tab:bdary.g}
\end{table}
\noindent Table~\ref{tab:lgtab}-Table~\ref{tab:igtab} demonstrate
survival probabilities ($\rho_n(x)$) in finite time and the
corresponding lower estimates ($\underline{\rho_n(x)}$) using multiple-Chebyshev inequalities with different moments of
  $Z_n$ depending on $x$.  These Tables demonstrate the way the survival
  probabilities get lower as time $n$ increases when $x$ fixed.  They also display how much higher the survival
  probabilities get for
  each time $n$ as an initial stock $x$ increases.  In
  addition to these observations, the estimates, $\underline{\rho_n(x)}$, by multiple-Chebyshev inequalities with different orders of moments of $Z_n$
  depending on $x$ become closer to $\rho_n(x)$ as $x$ becomes larger for each $n$.  In
  particular, one can observe that the survival probabilities $\rho_n(x)$ for
  Pareto distribution are more closely approximated by the estimates $\underline{\rho_n(x)}$ than those for the other two distributions as $x$
  gets larger for each $n$.  Moreover, the lower estimates
  $\underline{\rho_n(x)}$ for lognormal distribution approximate the
  survival probabilities $\rho_n(x)$ slightly better than those for
  gamma distribution do.  It turns out that the lower estimates of
  $\rho_n(x)$ by multiple-Chebyshev
  inequalities with different orders of moments of $Z_n$ depending on
  $x$ for Pareto distribution, lognormal, and gamma in
  this order show good performance of approximating the true probabilities.

\begin{table}[ht]
\centering
\begin{tabular}{rrrrrrrrr}
  \hline
$x$ & $\rho_{l,3}(x)$ & $\rho_{3}(x)$ & $\rho_{l,5}(x)$ & $\rho_{5}(x)$ & $\rho_{l,10}(x)$ & $\rho_{10}(x)$ & $\rho_{l,20}(x)$ & $\rho_{20}(x)$ \\ 
  \hline
3.5 & 0.2202 & 0.7530 & 0.0000 & 0.3633 & 0.0000 & 0.1290 & 0.0000 & 0.0907 \\ 
  7.5 & 0.9907 & 0.9997 & 0.9257 & 0.9927 & 0.5396 & 0.8890 & 0.3027 & 0.8193 \\ 
  9.5 &  &  & 0.9806 & 0.9997 & 0.8190 & 0.9700 & 0.6258 & 0.9280 \\ 
  12.5 &  &  & 0.9957 & 1.0000 & 0.9555 & 0.9963 & 0.8605 & 0.9807 \\ 
   \hline
\end{tabular}
\caption{The lower estimates of $\rho_n(x)$ using multiple-Chebyshev inequalities ($\rho_{l,n}(x)$) vs the survival probabilities ($\rho_n(x)$) for $\ln\epsilon_1\sim$ N(0.2146,0.0645)} 
\label{tab:lgtab}
\end{table}

\begin{table}[ht]
\centering
\begin{tabular}{rrrrrrrrr}
  \hline
$x$ & $\rho_{l,3}(x)$ & $\rho_{3}(x)$ & $\rho_{l,5}(x)$ & $\rho_{5}(x)$ & $\rho_{l,10}(x)$ & $\rho_{10}(x)$ & $\rho_{l,20}(x)$ & $\rho_{20}(x)$ \\ 
  \hline
3.5 & 0.2296 & 0.7070 & 0.0000 & 0.3557 & 0.0000 & 0.1713 & 0.0000 & 0.1393 \\ 
  7.5 & 1.0000 & 1.0000 & 0.9976 & 1.0000 & 0.5718 & 0.8783 & 0.3027 & 0.7930 \\ 
  9.5 &  &  &  &  & 0.8972 & 0.9770 & 0.6517 & 0.9323 \\ 
  12.5 &  &  &  &  & 0.9961 & 0.9990 & 0.9135 & 0.9853 \\ 
   \hline
\end{tabular}
\caption{The lower estimates of $\rho_n(x)$ using multiple-Chebyshev inequalities ($\rho_{l,n}(x)$) vs the survival probabilities ($\rho_n(x)$) for $\epsilon_1\sim$ Pareto($\beta$=3,$k$=0.9)} 
\label{tab:prtab}
\end{table}
\begin{table}[ht]
\centering
\begin{tabular}{rrrrrrrrr}
  \hline
$x$ & $\rho_{l,3}(x)$ & $\rho_{3}(x)$ & $\rho_{l,5}(x)$ & $\rho_{5}(x)$ & $\rho_{l,10}(x)$ & $\rho_{10}(x)$ & $\rho_{l,20}(x)$ & $\rho_{20}(x)$ \\ 
  \hline
3.5 & 0.2202 & 0.7860 & 0.0000 & 0.3653 & 0.0000 & 0.1293 & 0.0000 & 0.0780 \\ 
  7.5 & 0.9901 & 0.9997 & 0.9192 & 0.9887 & 0.5347 & 0.9133 & 0.3027 & 0.8267 \\ 
  9.5 &  &  & 0.9848 & 0.9983 & 0.8102 & 0.9767 & 0.6206 & 0.9293 \\ 
  12.5 &  &  & 0.9983 & 1.0000 & 0.9490 & 0.9957 & 0.8508 & 0.9777 \\ 
   \hline
\end{tabular}
\caption{The lower estimates of $\rho_n(x)$ using multiple-Chebyshev inequalities ($\rho_{l,n}(x)$) vs the survival probabilities ($\rho_n(x)$) for $\epsilon_1\sim$ $\Gamma$($\alpha$=17,$\theta$=13.3333)} 
\label{tab:igtab}
\end{table}
\end{exmp}

\end{document}